\documentclass{llncs}
\usepackage[utf8]{inputenc}
\usepackage[T1]{fontenc}
\usepackage[scaled=0.8]{beramono}
\usepackage{amsmath}
\usepackage{amssymb}
\usepackage{stmaryrd}
\usepackage{mathtools}
\usepackage{proof}
\usepackage{todonotes}
\usepackage{subcaption}
\usepackage{wrapfig}
\usepackage{sidecap}
\usepackage{booktabs}
\usepackage{multirow}
\usepackage{lipsum}

\usepackage{tabto}

\usepackage{listings}
\usepackage{newunicodechar}

\title{A Hoare Logic with \\ Regular Behavioral Specifications}
\author{Gidon Ernst \inst{1}
       \and Alexander Knapp \inst{2}
       \and Toby Murray \inst{3}}
\institute{LMU Munich, Germany, \email{gidon.ernst@lmu.de}
        \and Augsburg University, Germany, \email{knapp@informatik.uni-augsburg.de}
        \and University of Melbourne, Australia, \email{toby.murray@unimelb.edu.au}}

\usepackage[numbers,sort&compress]{natbib}

\usepackage{comment}
\usepackage{hyperref}
\hypersetup{hidelinks,
    colorlinks=true,
    allcolors=blue,
    pdfstartview=Fit,
    pdfusetitle=true,
    breaklinks=true}

\usepackage[capitalise]{cleveref}

\Crefname{equation}{Eq.}{Eqs.}
\Crefname{figure}{Fig.}{Figs.}
\Crefname{tabular}{Table}{Tabs.}
\Crefname{example}{Ex.}{Exs.}
\Crefname{section}{Sec.}{Sects.}
\Crefname{corollary}{Cor.}{Cors.}
\Crefname{proposition}{Prop.}{Props.}
\Crefname{theorem}{Thm.}{Thms.}
\Crefname{lemma}{Lem.}{Lems.}
\Crefname{algorithm}{Alg.}{Algs.}
\Crefname{definition}{Def.}{Defs.}

\newcommand{\red}[1]{{\color{red}#1}}

\newcommand{\code}[1]{\texttt{\upshape #1}}
\newcommand{\event}[1]{\code{#1}}
\newcommand{\hoare}[3]{\{\,#1\,\}~#2~\{\,#3\,\}}
\newcommand{\trace}[1]{\langle\,#1\,\rangle}
\newcommand{\derive}[2]{\delta_{#1} #2}
\newcommand{\step}[1]{\xrightarrow{~#1~}}
\newcommand{\steps}[1]{\mathrel{\xrightarrow{~#1~}\!\!{^*}}}
\renewcommand{\L}{\mathcal{L}}

\newcommand{\concat}{{\cdot}}

\newcommand{\aborted}{\mathbf{abort}}
\newcommand{\stopped}[1]{(\mathbf{stop}\ #1)}
\newcommand{\running}[2]{(\mathbf{run}\ #1\colon #2)}

\newcommand{\SPEC}[3]{#1\colon [#2 \leadsto #3]}

\newcommand{\DO}{\code{do}}

\newcommand{\SKIP}{\code{skip}}
\newcommand{\WHILE}{\code{while}}
\newcommand{\IF}{\code{if}}
\newcommand{\THEN}{\code{then}}
\newcommand{\ELSE}{\code{else}}

\newcommand{\END}{\code{end}}
\newcommand{\EMIT}{\code{emit}}

\begin{document}
\maketitle

\begin{abstract}
We present a Hoare logic that extends program specifications
with regular expressions that capture behaviors
in terms of sequences of events that arise during the execution.
The idea is similar to session types or process-like behavioral contracts,
two currently popular research directions.
The approach presented here strikes a particular balance between expressiveness and proof automation,
notably, it can capture interesting sequential behavior across multiple iterations of loops.
The approach is modular and integrates well with autoactive deductive verification tools.
We describe and demonstrate our prototype implementation in SecC using two case studies:
A matcher for E-Mail addresses
and a specification of the game steps in the VerifyThis Casino challenge.
\end{abstract}

\section{Introduction}
\label{sec:introduction}

Context of this work is the aim to specify the state machine logic of programs that have a control state
and in which the sequence of events matters.
A typical example are stateful protocols (e.g. the TCP handshake) and device drivers.
One motivation for this paper in particular is the Casino case study
from the ongoing VerifyThis discussion series since 2021.%
    \footnote{\url{https://verifythis.github.io/casino/}}
Goal of these discussions is to bring different communities together
and bridge specifications using contract-based mechanisms typically found in deductive verification tools and the automata-based approaches of model-checking.
Of particular interest was the integration of
automata-like specifications and models with deductive program verification tools.
Existing approaches include an embedding of a highly expressive process-calculus language
into the Separation Logic assertions~\cite{penninckx2019specifying,sprenger2020igloo,oortwijn2020abstraction}.
At the end of the spectrum are approaches where the C program itself
is abstracted such that it can be checked directly using a software model checker.
Such techniques can cope with temporal properties \cite{dietsch2015fairness,urban2018abstract}.
Of course, the research area of attributing and verying temporal behavior of systems
has a long history~\cite{pnueli1977temporal,lamport1994temporal,alur2010temporal}.

In this paper we discuss a particular point in the design space:
We embed declarative behavioral specification of event traces, written as regular expressions,
into the pre- and postconditions of a Hoare-like logic.
The approach combines strong reasoning about the functional correctness of programs
with a light-weight and fully automatic integration of behavioral aspects into contracts.

% It turns out that the naive approach
% to take as the specification simply the repetition of the loop body's traces is not that useful,
% as it cannot express sequential dependencies of the occurrence of events.
% This lack of expresiveness arises already for programs that encode simple state machines, including that of the Casino case study, and the motivational example shown in \cref{sec:motivation} which alternates between emitting two different events.
% In order to address this, we employ regular expression specifications that can be conditioned on arbitrary predicate logic formulas, such that the verification approach becomes complete wrt. behaviors that can be expressed by regular languages.

As shown in \cref{sec:motivation}, it turns out that the straight-forward approach to
specifying loops just by using the repetition operator is not expressive enough,
as it fails to capture sequential behavior across loop iterations,
which necessitates to include logical conditions into the annotation language.
In \cref{sec:approach}, we develop a Hoare-like logic with judgements $\hoare{P \mid U}{c}{Q \mid V}$, where in addition to the ordinary pre-/postcondition pair $P,Q$ for command~$c$, we have two regular trace specifications $U$ and $V$ that capture
which events may be emitted during the execution of the program.
The logic is proved sound in the Isabelle/HOL proof assistant and we suggest a corrresponding completeness result.
The clear advantage of regular expressions over more complex languages
is that the additional verification obligations can be decided automatically with little additional effort in the verification.
We have implemented the approach in the deductive verification tool SecC~\cite{ernst:cav2019},
as discussed with details on this automation in \cref{sec:implementation}.
We illustrate and discuss in depth the approach with two case studies:
A matcher for E-Mail addresses (\cref{sec:lex}) and the Casino challenge (\cref{sec:casino}).
% In \cref{sec:case-studies},

The \emph{contribution} of this paper is therefore
to propose this approach as natural point in an active research area
(\cref{sec:related})
and to showcase its merits and limitations
on practical examples, backed by a tool implementation.

%%% Local Variables:
%%% mode: LaTeX
%%% mode: TeX-PDF
%%% TeX-source-correlate-mode: t
%%% TeX-master: "main.tex"
%%% ispell-local-dictionary: "en_GB"
%%% End:

\section{Motivation}
\label{sec:motivation}

As an example consider a loop that generates a sequence
of alternating events $\event{even}$ and $\event{odd}$,
where the loop test nondeterministically terminates the loop
but only in a state when $b$ is \code{false}, expressed as pseudocode:

\begin{wrapfigure}{l}{0.33\textwidth}
\vspace*{-0.8cm}
\hrule
\vspace*{-0.3cm}
\begin{align*}
& b \coloneqq \code{false} \\[-.6ex]
& \WHILE\ \code{*} \lor b \ \DO \\[-.6ex]
& \quad \IF\ b\ \THEN\ \EMIT\ \event{odd}\\[-.6ex]
& \quad \phantom{\IF\ b\ } \ELSE\ \EMIT\ \event{even} \\[-.6ex]
& \quad b \coloneqq \lnot b \\[-.6ex]
& \END
\end{align*}
\vspace*{-0.8cm}
\caption{Event alternation}
\vspace*{+0.1cm}
\hrule
\vspace*{-0.8cm}
\label{fig:motivation:even-odd}
\end{wrapfigure}
The specification command $\EMIT~a$ demarks
the occurrence of an event during program execution,
and we are interested in specifying the behavior of programs
in terms of the possible event traces that can occur at runtime.
Such events are sometimes attributed to certain steps of the program's semantics,
but here we just keep them abstract and assume that the right events have been placed at appropriate locations in the source code in terms of $\EMIT$ statements.
%
% In this paper we will use regular expressions because of their computational tractability,
% notably, we will make use of the fact that the language inclusion problem is decidable.
% However, in practice, other specification languages such as context-free grammars,
% linear temporal logic, or more general process-based languages are possible,
% supported by heuristics and/or semi-automated reasoning approaches.
The behavior of the above program is captured nicely by
\begin{align}
\textbf{program specification} &
    &
(\event{even} \concat \event{odd})^*
    \label{spec:ex1}
\end{align}
where $\concat$ denotes concatenation and $^*$ denotes repetition.

The program is indeed correct with respect to this specification,
informally because the loop body is executed for the first time with $\lnot b$,
thus emitting an $\event{even}$ event first,
the loop alternates between the two events,
and terminates only after no event or after an $\event{odd}$ event has been emitted.
To make this argument formally precise, we aim at a verification approach that
1) is complete with respect to regular languages and
2) integrates well into existing modular approaches that can be presented in the style of Hoare logic.
% Notably, we aim at specifying the traces generated by a procedure modularly
% as part of their respective contract,
% and we aim to support reasoning about loops in terms of a specification of the loop body, similarly to loop invariants.

A simple but straightforward idea to specify the traces of a while loop
is to repeat whatever the body produces, using the star operator.
In the example above, the behavior of loop body can be abstracted by 
the regular expression $(\event{even} \mid \event{odd})$
from which we can conclude that the loops adheres to $(\event{even} \mid \event{odd})^*$.
However, this naive approach cannot capture the
alternation of \event{even} and \event{odd} events as produced by the loop.
This coincides with the fact that the finite automaton
that accepts the same language as~\eqref{spec:ex1} has \emph{two} states,
thus, we need \emph{two} regular expressions
to describe traces with respect to the current value of~$b$.

Just as with the functional component of loop specifications,
we can choose whether we want to describe the events observed so far (analogously to an invariant),
or alternatively whether we want to specify
the traces that are still permitted to occur (analogously to a loop postcondition or summary~\cite{hehner2005specified,tuerk2010local,ernst:vmcai2022}).
For this example, the alternatives are as follows
\begin{align*}
\textbf{loop invariant} &
& (\event{even} \concat \event{odd})^* \
    &\IF\ \lnot b
& (\event{even} \concat \event{odd})^* \concat \event{even} \
    &\IF\ b
    \\
\textbf{loop summary} &
& (\event{even} \concat \event{odd})^* \
    &\IF\ \lnot b
& \event{odd} \concat (\event{even} \concat \event{odd})^* \
    &\IF\ b
\end{align*}
The regular expressions are conditioned upon a formula which depends on the value of $b$ in some arbitrary intermediate state encountered at the loop head.
If~$b$ is currently true, then the invariant expresses that this current state
has been reached by a trace prefix that ends in an $\event{even}$ event,
such that the subsequent iteration \emph{concatenating} an $\event{odd}$ to the end of this regular expression will again give~\eqref{spec:ex1}.
Conversely, the loop summary in this case makes it explicit that the next
event expected is an $\event{odd}$ event.
In that regard, the two approaches are dual to one each other and perfectly symmetric for regular expressions because the repetition operator can be expressed as both a left-fold and as a right-fold recursively (which is not the case for context-free grammars, in general).
For now we will focus on the invariant approach
as it is conceptually more similar to how functional correctness is typically established.
% and also the underlying meta-theory appears to be simpler.
% We remark, however, that the initial implementation in SecC was actually based on the second approach and we will contrast how the proof obligations work out later.
% Soundness of the summary-like approach seems to require
% a backwards simulation argument that we have yet to put together,
% whereas the soundness proof for the invariant-like approach is completely straight-forward.
% We conjecture that both approaches are equivalent in expressive power,
% though we leave this matter open for future work.

%%% Local Variables:
%%% mode: LaTeX
%%% mode: TeX-PDF
%%% TeX-source-correlate-mode: t
%%% TeX-master: "main.tex"
%%% ispell-local-dictionary: "en_GB"
%%% End:

\section{Approach}
\label{sec:approach}

In this section, we show an extension of Hoare logic that integrates
trace specifications in terms of regular expressions with deductive proofs,
where program correctness is expressed as judgements of the form $\hoare{P \mid U}{c}{Q \mid V}$
with respect to pre-/postconditions $P$, $Q$ and a command $c$.
In addition, there are two conditional regular expressions $U$ and $V$,
subsequently called regular behavioral specifications,
denoting the trace prefix and resulting trace, respectively.
% This section continues with a formal definition of the trace specifications in \cref{sec:regular},
% followed by the semantic definition of judgements in \cref{sec:commands}
% with respect to a simple imperative programming language.
% The corresponding proof rules are presented \cref{sec:logic} alongside some examples.

\vspace*{-.4ex}
\paragraph{Preliminaries.} Conditions $P,Q$ and $\phi$ as well as binary relations~$R$
are represented syntactically,
where the latter are formulas involving also primed variables, e.g., to denote successor states and nondeterministic transitions.
For example $x' > x$ encodes that the value of program variable $x$ is strictly increased by a nondeterministic value.
Ascribing a prime symbol to a predicate as in~$P'$ is understood to prime all of its free variables, similarly $U'$ is~$U$ with all free variables in all conditions
replaced by their primed counterpart.
Semantically, formulas can be evaluated over states~$s \in S$,
written $P_s$ for instance, to result in a semantic truth value.
Usually, we treat primed variables just as other variables,
however, occasionally we denote by $R_{s,s'}$ the evaluation of relation $R$
taking the unprimed variables from state~$s$
and the primed variables from $s'$.

\subsection{Regular Behavioral Specifications}
\label{sec:regular}

Plain regular expressions~$u,v,w$ consist of
the empty language, the empty word, symbols of an alphabet $a \in A$,
sequential composition, choice, and repetition:
\begin{align*}
u   & \Coloneqq \varnothing \mid \epsilon \mid a \mid (u \concat v) \mid (u {\mid} v) \mid u^*
\end{align*}
The language $\L(u)$ of an expression~$u$ is defined as usual
as a set of words which are finite sequences of symbols, representing behavioral traces $\tau = \trace{a_1, \ldots, a_n}$ here.
A regular expression~$u$ is called \emph{nullable} if its language contains the empty word, i.e., $\trace{} \in \L(u)$;
and~$u$ is called \emph{empty} if its language is the empty set, i.e., $\L(u) = \varnothing$.
Language inclusion and equivalence are defined as follows
\begin{align*}
u \sqsubseteq v
    & \iff \L(u) \subseteq \L(v) &
u \equiv v
    & \iff \L(u) = \L(v) &
\end{align*}
such that $u$ is nullable if and only if $\epsilon \sqsubseteq u$,
and $u$ is empty if $u \equiv \varnothing$.
Nullability and emptiness can be checked efficiently by a simple recursion over the syntax.

\begin{definition}[Regular Behavioral Specification]
A regular behavioral specification~$U,V,W$
is a state-dependent choice between plain regular expressions~$u_i$:
% each conditioned with a syntactic formula~$\phi_i$:
\begin{align*}
U \Coloneqq ~~u_1\ \IF\ \phi_1 ~~ \mid ~ \cdots ~ \mid ~~ u_n\ \IF\ \phi_n
\end{align*}
\end{definition}
We tacitly interpret a plain regular expression
$u\ \IF\ \mathit{true}$
with a trivial guard as such a specification.
Evaluation in a particular state simply collects
the options with a valid test, and discards all others:
\begin{align*}
(u_1\ \IF\ \phi_1 \mid \cdots \mid u_n\ \IF\ \phi_n)_s
    &= 
(u_1\ \IF\ \phi_1)_s \mid \cdots \mid (u_n\ \IF\ \phi_n)_s
    & \text{where}
    \\
(u\ \IF\ \phi)_s
    &= u \text{ if } \phi_s \quad \text{ and }  \quad
(u\ \IF\ \phi)_s
    = \varnothing \text{ otherwise}
\end{align*}
Later, for rule \textsc{Frame} in \cref{sec:logic},
we refer to sequential composition of such specifications which semantically obeys
$(U \concat V)_s = U_s \concat V_s$
but we refrain here from giving a general syntactic account for brevity,
and because this construct may be somewhat
misleading as $U$ and $V$ are evaluated in the \emph{same} state,
even though sequential composition suggests that $V$ might occur \emph{after} $U$,
and maybe after a state change in the program.
However, we can freely move state-independent regular expression fragments
in and out of this composition as in $(u \concat V)_s = u \concat (V_s)$,
which will be the use case in the implementation (\cref{sec:implementation}).

In order to reason about these conditionals,
we reflect language inclusion and equivalence
as predicates with with the following semantics
\begin{align*}
(U \sqsubseteq V)_s
    & \iff U_s \sqsubseteq V_s
    &
(U \equiv V)_s
    & \iff U_s \equiv V_s
\end{align*}
such that, e.g., $P \implies U \sqsubseteq V$ can be regarded as a logical formula.
This will be useful in particular to express a strong consequence rule in \cref{sec:logic}.

\begin{comment}

A regular expression specification~$U$ can be reduced to a plain regular expression wrt. $P$, written $U_P$,
by keeping only those $u_i$ from each option $u_i\ \IF\ \phi_i$ where $P \land \phi_i$ is satisfiable,
which can be checked for example by an SMT Solver.
Then, $P \implies U \sqsubseteq V$ holds if (but not only if) $U_P \sqsubseteq U_Q$,
which can be checked algorithmically as stated above.
\red{only needed for summaries}
$U$ is called \emph{complete}, if $\phi_1 \land \cdots \land \phi_n$ is a tautology,
and we define the \emph{standard completion} of (any)~$U$ as
the addition by an additional term
\[ U \mid \epsilon\ \IF\ \lnot\phi_1 \land \cdots \land \lnot\phi_n \]
If $U$ is already complete this extension preserves the meaning of~$U$.
On the other hand, if $U$ is the empty choice, standard completion
yields just~$\epsilon$ a the default behavior of programs that emit no events.
Completion is crucial to ensure that program paths are not discarded out accidentally, because the specification holds vacously.%
    \footnote{Anecdote: This did in fact occur during development of tool support,
              because the verification of some regression tests suddenly succeeded
              instead of producing a proof failure as it would have been expected.}
\end{comment}

\subsection{Programs and Behavioral Correctness}
\label{sec:commands}

Imperative commands~$c$ are formed by the grammar shown below,
comprising specification statements~\cite{morgan1988specification}
and the usual composition constructs:
\begin{align*}
c \Coloneqq \SPEC{\vec x}{G}{R \mid U} \mid c_1;c_2 \mid
    \IF\ t\ \THEN\ c_1\ \ELSE\ c_2 \mid
    \WHILE\ t\ \DO\ c \mid \cdots
\end{align*}
Atomic commands are subsumed by specification statements $\SPEC{\vec x}{G}{R \mid U}$, extended by a regular behavior.
These can abstract over an arbitrary program fragment:
Provided that guard~$G$ holds in a given state,
this command takes a nondeterministic transition by modifying the variables~$\vec x$
according to a transition relation~$R$, and by emitting a trace of the language of $U$.
The specification statement can encode some more conventional
constructs like $\SKIP$ and assignments,
but also the emission of a single event as shown earlier, $\EMIT\ a \equiv \SPEC{-}{\code{true}}{\code{true} \mid a}$, with no modified variables and no constraints on the transition.

Commands~$c$ execute according to a natural big-step semantics
$k \steps{\tau} k'$
from an initial configuration $\running{s}{c}$ with state~$s$
to a final configuration that is either $\stopped{s'}$ (regular termination in final state $s'$) or $\aborted$ (subsuming runtime errors).
The sequence of events that have happened during this particular execution
is annotated as a trace~$\tau$.
The definition of the rules governing $k \steps{\tau} k'$ are entirely standard,
except perhaps for the specification statement:
\begin{align*}
\running{s}{\SPEC{\vec x}{G}{R \mid U}} &\step{\tau} \aborted
    && \hspace*{-.5em}\text{if not } G_s \text{ and } \tau \text{ arbitrary\footnotemark}
    \\[-.5ex]
\running{s}{\SPEC{\vec x}{G}{R \mid U}} &\step{\tau} \stopped{s'}
    && \hspace*{-.5em}\text{if } G_s \land R_{s,s'} \land \tau \in U_s \text{ and } s' = s[\vec x \mapsto \vec v]
    % Note: using \land because it looks better than comma
\end{align*}
where $s[\vec x \mapsto \vec v]$ denotes the modified state
in which the variables $\vec x$ have been updated to some arbitrary new values $\vec v$,
leaving all other variables unchanged.

As a consequence, the derived statement $\EMIT~a$ behaves as expected
\begin{align*}
\running{s}{\EMIT~a} &\step{\trace{a}} \stopped{s}
\end{align*}
In this paper we do not address the issue of (non-)termination---%
diverging runs are simply not generated by this semantics,
and the resulting logic will express functional correctness as well as
behavioral correctness with respect to the trace for terminating executions only.
\footnotetext{The arbitrary trace~$\tau$ reflects that a diverging program could have any effect.}

Judgements
$\hoare{P \mid U}{c}{Q \mid V}$
comprise the usual constituents,
precondition~$P$, command~$c$, and postcondition~$Q$,
as well as the (state-dependent) regular expressions specification~$U$
over the alphabet $A$ collecting possible trace prefixes seen so far,
and~$V$ constraining how these may be extended by executing~$c$.
\begin{definition}[Valid  Hoare Triples]
    \label{def:hoare}
A Hoare triple
$ \hoare{P \mid U}{c}{Q \mid V}$
 is \emph{valid}, if for all
$s \in S$ with traces~$\tau$, and for all configurations~$k'$
\begin{align*}
& P_s \text{ and } \running{s}{c} \steps{\tau} k' \text{ implies }\\
& \qquad 
\text{there is } s' \text{ with } k' = \stopped{s'}
\text{ and } Q_{s'}
\text{ and } \L(U_{s} \concat \tau) \subseteq L(V_{s'})
\end{align*}
\end{definition}
Note that $P$ and $U$ are evaluated in the pre-state~$s$,
whereas $Q$ and $V$ are evaluated in the post-state~$s'$,
and that $k' \neq \aborted$.

\begin{comment}
We require that specifications $U$ can be presented finitely,
e.g., as an enumeration $[u_1\ \IF\ \phi_1 \mid \cdots \mid u_n\ \IF\ \phi_n]$,
such that $\phi_1 \land \cdots \land \phi_n$ is a tautology.
In the tool, we always ensure this by completing a given specification
by an additional component $\epsilon\ \IF\ \lnot(\phi_1 \land \cdots \land \phi_n)$.
This ensures that absence of annotations specifies enfores empty trace $\trace{}$,
and it helps with ensuring that specifications do not denote the empty language in practice. Semantically:

We freely mix regular expression notation such as concatenation and repetition
between regular expression specifications~$U$ (uppercase), regular regular expressions~$u$ (lowercase), and traces~$\tau$.

By $U \sqsubseteq V$ we denote that $\L(U_s) \subseteq \L(V_s)$ for all~$s$.
We later present a very simple but occasionally inefficient algorithm to check this and that avoids an explicit construction of automata.
\end{comment}

\subsection{Hoare Logic Proof Rules}
\label{sec:logic}

% We point out that $U$ is consistently evaluated in initial states~$s$,
% % whereas $V$ is evaluated in final states~$s'$.
% The requirement that $U_s$ denotes a nonempty language is convenient
% for the soundness proof of the invariant rule (see below \red{and LaTeX comment}).
% % For all non-empty traces, we can derive this fact from the last condition
% after a zero-step execution, provided we have the precondition P for some state.
% This does not work with the premise of the loop rule, which requires the positive loop test,
% but the loop could exit immediately.

The proof rules are standard with respect to the functional correctness aspects
captured by the pre-/postcondition.
Regarding the trace specifications, there are some aspects that are worth discussing.

Typically, traces that occur during execution are simply added to the regular expression specification in the precondition. For example, the derived rule for the \code{emit} specification statemement is just
\begin{align*}
\infer[\textsc{Emit}]
    {\hoare{P \mid U}{\EMIT~a}{P \mid U \concat a}}{}
\end{align*}
It follows from the more general rule specification statement,
which may be thought of as mirroring a procedure call that is verified modularly.
The statement can be executed whenever guard $G$ follows from the current precondition~$P$.
Given then that binary relation~$R$ between unprimed and primed variables holds,
we need to establish $Q$ in that successor state,
and we also need to ensure that the extension of the pre-traces~$U$ by any trace of~$W$ is covered by $V$.
\begin{align*}
\infer[\textsc{Spec}]
    {\hoare{P \mid U}{\SPEC{x}{G}{R \mid W}}{Q \mid V}}
    {P \implies G &&
     P \land R \implies Q' \land (U \concat W \sqsubseteq V')}
\end{align*}
Recall that $Q'$ and $V'$ reference the successor state,
denoted in the logical formula in terms of the primed variables constrained by~$R$,
whereas~$U$ and~$W$ reference the state before executing the program
in accordance with the semantics.

\begin{example}
\label{ex:spec}
We can represent the loop body of our motivating example in \cref{fig:motivation:even-odd}
  by a single specification statement for demonstration purposes here
\begin{align*}
\SPEC{b}{\code{true}}{b' = \lnot b \mid a}
    \quad \text{ where }
        a = \begin{cases}
        \event{even},& \text{ if } \lnot b \\
        \event{odd}, &  \text{ otherwise}
        \end{cases}
\end{align*}
Recall the invariant-like characterization of traces from \cref{sec:motivation},
formalized as a regular expression specification $U(b)$ over program variable~$b$:
\begin{align*}
U(b) \quad \equiv \quad
(\event{even} \concat \event{odd})^*\ \IF\ \lnot b
    ~\mid~
(\event{even} \concat \event{odd})^* \concat \event{even}\ \IF\ b
\end{align*}

The part of the instantiated precondition of rule \textsc{Spec}
that relates the pre-traces to the post-traces takes $V = U$
to re-establish the invariant after executing the body,
which leads to the following proof obligation
\begin{align*}
b' = \lnot b \implies U(b) \concat a \sqsubseteq U(b')
\end{align*}
This can be reduced by case analysis where $U(b)$ and $U(b')$
pick the the opposite regular expression, respectively,
where the most recently emitted event~$a$ is underlined:
\begin{align*}
(\event{even}\concat\event{odd})^* \concat \underline{\event{even}}
    &\sqsubseteq
        (\event{even}\concat\event{odd})^* \concat \event{even}
    && \text{ if } \lnot b
    \\
(\event{even}\concat \event{odd})^* \concat \event{even} \concat \underline{\event{odd}}
    &\sqsubseteq (\event{even}\concat \event{odd})^*
    && \text{ otherwise }
\end{align*}
Clearly, both conditions are satisfied.
    \hfill $\heartsuit$
\end{example}

The consequence rule exhibits the usual co-/contravariance duality:
We may wish to conduct the proof with a weaker precondition~$P_2$
and larger set of traces seen so far~$U_2$
to establish a stronger postcondition~$Q_2$ than necessary
together with a more precise set of extended traces~$V_2$.
\begin{gather*}
\infer
    {\hoare{P_1 \mid U_1}{c}{Q_1 \mid V_1}}
    {P_1 \implies P_2 \land (U_1 \sqsubseteq U_2) & 
     \hoare{P_2 \mid U_2}{c}{Q_2 \mid V_2} & 
     Q_2 \implies Q_1 \land (V_2 \sqsubseteq V_1)}
   \\[-0.55cm]
   \hspace*{10.9cm}\textsc{Conseq}
\end{gather*}
This rule feeds contextual information
from predicate logic formulas into the inclusion conditions of regular language specification,
just as we have seen for the \textsc{Spec} rule in \cref{ex:spec}.
Note that this rule makes up for not having defined sequential concatenation
of trace specifications:
We can manually re-arrange regular expressions with respect to their outer conditionals.
The consequence rule can also feed information about the outcome of the test
into the premises of the \textsc{If} rule (not shown).

The rule for sequential composition just chains the intermediate conditions
and thereby avoids to refer to sequential composition of trace specifications:
\begin{align*}
\infer[\textsc{Seq}]
    {\hoare{P \mid U}{c_1;c_2}{R \mid W}}
    {\hoare{P \mid U}{c_1}{Q \mid V} &&
     \hoare{Q \mid V}{c_2}{R \mid W}}
\end{align*}

As we have motivated by the examples above, the rule for while loops
employs in addition to an ordinary invariant~$I$ a regular expression specification~$U$
that characterizes the traces observed when executing the loop.
\begin{align*}
\infer[\textsc{While}]
    {\hoare{I \mid U}{\WHILE\ t\ \DO\ B}{\lnot t \land I \mid U}}
    {\hoare{t \land I \mid U}{B}{I \mid U}}
\end{align*}
The same~$U$ occurs for all traces in the rule in four places.
This reflects the fact that the loop specification must already be given in a
closed form that is generalized with respect to an arbitrary number of iterations.
The evaluation of~$U$
is still relative to the context for that occurrence,
specifically, with the postcondition~$\lnot t \land I$ in the conclusion
we can narrow down~$U$ to those cases that actually match the exit state.

\begin{example}
\label{ex:loop}
In \cref{sec:motivation}, we had specified that the loop must not exit
when~$b$ is true, such that $U(b)$ simplifies to the desired
\emph{overall} behavior of the program with properly paired even and odd events
$(\event{even}\concat\event{odd})^*$.
    \hfill $\heartsuit$
\end{example}

The rules seen so far do not allow one to ever ``forget'' any prefix of the trace seen so far.
This means that the trace specification of loops is not really modular---%
$U$~must maintain any events that have occurred in the past.
The following framing rule therefore allows one to put a subpart of the execution into
the context of a prefix~$W$ that captures the past up to that point.
This rule encodes that the program execution itself cannot depend on past events,
because traces have no manifestation at runtime.
\begin{align*}
\infer[\textsc{Frame}]
    {\hoare{P \mid W \concat U}{c}{Q \mid W \concat V}}
    {P \land Q' \implies W \equiv W' &&
     \hoare{P \mid U}{c}{R \mid V}}
\end{align*}
Just as with the frame rule in Separation Logic~\cite{reynolds2002separation}
there is a side-condition that the frame~$W$ is independent of the program execution,
which we can express by an additional premise as shown.
This premise is trivially satisfied for plain regular expressions that
do not contain any state-dependent conditions
and our implementation uses this heavily as explained in~\cref{sec:implementation}.

With the help of this rule, we can justify the naive approach to
loops that just repeats any observation that can be made about the executions of the loop body
with out taking any sequencing constraints into account,
given here by a plain regular expression~$v$ to simplify the presentation.
\begin{align*}
\infer[\textsc{While}*]
    {\hoare{I \mid \epsilon}{\WHILE\ t\ \DO\ B}{\lnot t \land I \mid v^*}}
    {\hoare{t \land I \mid \epsilon}{B}{I \mid v}}
\end{align*}
It follows from rule \textsc{While} for invariant $U = v^*$,
and the frame rule for $W = v^*$ to justify the needed intermediate condition
    $\hoare{t \land I \mid v^*  \concat \epsilon}{B}{I \mid v^* \concat v }$
from the given premise of rule \textsc{While*}.
The rest of the justification is stitched together by rule \textsc{Conseq}
and the algebraic laws of the repetition operator.

As an example, taking $v = (\event{even} \mid \event{odd})$
we can derive the weaker characterization of the loop's behavior that
was mentioned in \cref{sec:motivation}.

\begin{theorem}[Soundness]
The rules presented in this paper are sound:
if the premises are valid according to \cref{def:hoare}
then the respective conclusion is also valid.
\end{theorem}
\begin{proof}
Mechanized in Isabelle/HOL, available for review%
\footnote{The mechanization will be made available permanently at a later date,
e.g., as an entry to the Archive of Formal Proofs at \url{https://www.isa-afp.org/}.} \\
\url{https://gist.github.com/gernst/6a156facbe402f6d7f8db8b6520c7d70}
    \qed
\end{proof}
We remark that the soundness proof does not in any way depend
on the fact that the trace specifications are regular,
or whether they can be presented as a finite enumeration.
% We furthermore state the following completeness property:
\begin{claim}[Completeness of Trace Annotations]
If there is an ordinary regular expression~$v$
that describes the traces of a given program~$c$,
i.e.,
$\hoare{P \mid \epsilon}{c}{Q \mid v}$ is valid,
we can find regular expression specifications to compose
this fact using the proof rules, including trace invariants for loops.
\end{claim}
\begin{proof}[Main Idea]
We can partition the actual state space of the program
into equivalence classes with respect to their trace prefixes at each program location.
Since the automaton corresponding to~$v$ has finitely many states,
we can describe them by a finite number of regular expressions
that are conditional upon a predicate logic formula characterizing the respective equivalence class.
    \qed
\end{proof}
In this paper, we do not make this claim formally precise,
but we think that the intuition is clear.
For example, the equivalence classes for the odd/even program from \cref{sec:motivation} would be given by~$\lnot b$ and~$b$ as expected.
The merit of this completeness property is that the proof rules themselves
already contain sufficient information to capture any regular behavior
relative to an adequate background theory in which the formulas in the conditions can be expressed.

%%% Local Variables:
%%% mode: LaTeX
%%% mode: TeX-PDF
%%% TeX-source-correlate-mode: t
%%% TeX-master: "main.tex"
%%% ispell-local-dictionary: "en_GB"
%%% End:

\section{Tool Support in SecC}
\label{sec:implementation}

We have implemented support for trace specifications in SecC,
which is an autoactive deductive program verifier for low-level concurrent C code with expressive security specifications.
It is based on the logic SecCSL~\cite{ernst:cav2019},
which for the purposes to this paper is analogous to standard concurrent separation logic.
The tool and case studies are publicly available at
\url{https://bitbucket.org/covern/secc}.

The tool follows the classical design of most deductive tools,
in which all C functions are specified modularly in terms of user-supplied contracts,
such that they can be verified in isolation to avoid the combinatorial path explosion
of interprocedural verification.
To that end, SecC supports program annotations for function contracts, loop invariants, auxiliary statements such as intermediate assertions and other proof hints.
It is possible, too, to specify logical functions and separation logic predicates,
as well as to encode mathematical proofs (e.g. by induction) as lemmas and lemma functions,
see e.g.~\cite{jacobs2010verifast}.

\subsection{Specification Syntax}

For this work, we extended the specification language by trace annotations.
We use the example from \cref{sec:motivation} as shown in \cref{fig:even-odd}
to explain the specification syntax in general as well as for the traces.

\begin{wrapfigure}{l}{0.5\textwidth}
\vspace*{-0.8cm}
\hrule
\vspace*{-0.1cm}
\begin{lstlisting}
int nondet();

void even_odd()
  _(trace (even odd)*)
{
  int b = 0;
  
  while(nondet() || b)
    _(trace (even odd)*      if !b)
    _(trace (even odd)* even if  b)
  {
    if(!b) { _(emit even) }
    else   { _(emit odd)  }
    
    b = !b;
  }
}
\end{lstlisting}
\vspace*{-0.4cm}
\caption{Even/odd example in SecC.}
\vspace*{+0.1cm}
\hrule
\vspace*{-0.8cm}
\label{fig:even-odd}
\end{wrapfigure}
In SecC, auxiliary code is wrapped inside \code{\_(...)}, an idiom taken from VCC.
This is used here for different purposes, foremost, in the contract of function \code{even\_odd()} to specify its behavior.
An annotation \code{\_({\bfseries trace} $U$)} for a function with body~$c$
corresponds to a Hoare triple $\hoare{\_ \mid \epsilon}{c}{\_ \mid U}$
(pre-/postconditions~$P$ and~$Q$ are absent in the example).
Inside the loop, the specification statement \code{\_({\bfseries emit} a)} generates the respective events.
It is the proof engineer's task to place these at the program locations of their interest---in the future we may support events that are implicitly generated, e.g., by function calls as in \cite{disney2011temporal}.
There is also a trace annotation for the loop, which showcases the concrete syntax for conditionals, analogous to the presentation earlier in this paper
albeit in SecC each choice is listed separatedly.
By convention, absence of a trace annotation, even for just a particular case,
enforces that no events may be emitted at all.%
    \footnote{We have omitted some annotations related to SecC's enforcement of absence of timing side-channels, which requires an invariant that \code{b} never contains any secret
such that we can branch on it,
similarly for the return value of function \code{nondet()}, shown at the top of the code listing, that models nondeterministic choice.}

\begin{comment}
In the example, there are some additional specifications related to
the timing-sensitive security guarantees that SecC enforces.
Notably, branching on a secret is not allowed.
In this example, this means that we have to show that the loop test
always evaluates to a non-secret value, which means we have to ensure
that both sides of the \code{||} operator satisfy this constraint.
We will subsequently omit this aspect in the examples.

It is specified to always return a non-secret value, i.e., the \code{\bfseries result} adheres to the \code{low} classification.
Similarly, the \code{b} local variable never contains a value that the attacker cannot infer from the source code:
Initially, it is initialized by constant value \code{0} and each iteration of the loop deterministically swaps it. Therefore, placing an invariant that the value of \code{b} remains public is justified.
\end{comment}

\subsection{Verification Engine}

In contrast to the Hoare-style rules of \cref{sec:approach},
which mirror those of the logic SecCSL behind the tool,
SecC operationalizes the proof rules by a forward symbolic execution algorithm,
first described here~\cite{berdine2005symbolic}, similarly to VeriFast~\cite{Jacobs11}
or the Silicon backend of Viper~\cite{muller2016viper}.
The engine traverses the program using execution states, consisting of
a store of symbolic variables, a path constraint, a symbolic representation of the heap, and for this work a trace prefix.
The design of the logic necessitates being careful about branching
not just for case distinctions in the program but also in the logic itself.
To address this, SecC always takes apart relevant case distinctions
and follows the branches individually, pruning those with an unsatisfiable path constraint.
The same principle is applied to the regular trace specifications,
which are broken apart into plain regular expressions eagerly,
such that SecC can make use of the rule \textsc{Frame},
e.g., to concatenate the possible trace extensions after a while-loop is dispatched modularly.
However, one has to be careful not to discard any execution branch vacuously
just because there is no trace specified for it.
Therefore, SecC always completes behavioral annotations with an empty trace
that is the negation of the disjunction of all other cases (which might of course be unsatisfiable).

At the end of each execution path, SecC needs to check
that the trace produced up to this point is covered by the behavioral specification
annotated to the surrounding context (a C function or a loop).
The inclusion check is implemented for plain regular expressions
by an approach similar to the one shown in~\cite{almeida2012deciding}.
It is based on the derivative~$\derive{a}{u}$ of~$u$ with respect to symbol~$a$
that captures the language of suffixes of~$u$ after a leading occurrence of~$a$,
i.e., $\L(\derive{a}{u}) = \{ \tau \mid a \concat \tau \in \L(u) \}$.
% 
% From the following unfolding property of language inclusion
% \begin{align*}
% u \sqsubseteq v \iff 
%     (u \text{ nullable} \implies v \text{ nullable})
% ~\text{and}~
%         (\forall\ a\in A.\ \derive{a}{u} \sqsubseteq \derive{a}{v})
% \end{align*}
From this definition we can derive an algorithmic check,
$\Gamma \vdash u \sqsubseteq v$, which maintains
a set~$\Gamma$ of expressions seen already
        and proceeds recursively as shown below.
Effectively, this algorithm computes in $\Gamma$ a simulation relation between~$u$ and~$v$
from the transitions that are possible over~$u$ via the derivative.
% If $u = \epsilon$, then we can just determine whether $v$ is nullable.
% If $v = \varnothing$ then $u$ must also denote the empty language.
% Otherwise, we $v$ must be nonempty and we unfold the inclusion with the help of the derivative
% for all symbols~$a \in A$ of the finite alphabet:
\begin{align*}
\Gamma \vdash u \sqsubseteq v
    & \iff
      \begin{cases}
      % v \text{ nullable},   & \text{if } u = \epsilon \\
      \mathit{true}, & \text{if } (u,v) \in \Gamma \\
      u \text{ empty}, & \text{if } v = \varnothing \\[2pt]
      \begin{array}{@{}l@{}}
      \Gamma \cup \{(u,v)\} \vdash \derive{a}{u} \sqsubseteq \derive{a}{v}, \forall\ a\in \mathit{first}(u),\\[-2pt]
      \quad\text{and } u \text{ nullable} \implies v \text{ nullable}
      % \footnotemark
      \end{array}
      & \text{otherwise}
      \end{cases}
\end{align*}
% \footnotetext{An earlier implementation of this algorithm in SecC
%               lacked the implication check wrt. nullability,
%               which caused some intuitive but inadequate specifications to pass verification.
%               We point these out in \cref{sec:case-studies}.
%               This defect was uncovered by mechanizing the algorithm in Isabelle/HOL,
%               which demonstrates the utility of computer supported proofs,
%               even when they are not applied at the code-level.}%
Here, $\mathit{first}(u) \subseteq A$ is an overapproximation of the set of first symbols of words in~$\L(u)$, which can be computed recursively from the syntax.
% Taking the entire alphabet~$A$ works, too, but the smaller $\mathit{first}(u)$ is more efficient.
Termination is ensured by unfolding each pair of expressions once only (first line).%
    \footnote{Note that it is crucial that the check whether a given pair of expressions is contained in $\Gamma$ already considers equivalence modulo reordering and duplicates of the choice operator, otherwise, one may keep accumulating larger and larger expressions.}
% We chose to include this algorithm here even if the ideas behind it are well-known,
% mainly to showcase how easy it is to support this feature in a language-based setting,
% i.e., without the need to convert regular expressions to automata.
% The approach appears to be reasonably efficient for the specifications we had looked at so far.
% It is straight-forward to support error reporting:
% given we keep track of the prefix trace that has lead through the recursion:
The second case catches when $\derive{a}{v}$ has become empty,
which corresponds to an event~$a$ of~$u$ that is not covered by the specification.
Conversely, in the last case when~$u$ is nullable but~$v$ is not,
some required events were missed.

% SecC reports these errors with some detail.
% For example, when the specification expects the sequence $\event{a} \concat \event{b}$
% but the program emits just~$\event{a}$, the tool outputs the following information,
% together with a symbolic path through the program (which is not shown here):
% \begin{verbatim}
%     unmatched trace:          a .
%     residual  specification:  b
% \end{verbatim}
% whereas when the program outputs $\event{a} \concat \event{c}^*$ instead, the error message reads as
% \begin{verbatim}
%     unmatched trace:          a . c
%     residual  trace:          c *
%     residual  specification:  b
% \end{verbatim}
% This feature turned out to be helpful even for the relatively large unfoldings
% of the regular expressions of the case studies, because errors typically manifest when either~$u$ or~$v$ is small.

%%% Local Variables:
%%% mode: LaTeX
%%% mode: TeX-PDF
%%% TeX-source-correlate-mode: t
%%% TeX-master: "main.tex"
%%% ispell-local-dictionary: "en_GB"
%%% End:

% \section{Case-Studies}
% \label{sec:case-studies}

% We present two case studies, followed by a discussion in \cref{sec:discussion}.

\section{Case Study: Regular Expression Matching}
\label{sec:lex}

A canonical case study for the specification approach of this paper
is to verify the state machine of a matcher that checks whether some input
adheres to a given concrete regular expression.
For instance, we aim at matching E-Mail addresses,
here with much simplified format \code{[a-z]+[@][a-z]+[.][a-z]+}
which expects a name part and domain part with a single dot,
where both parts are separated by an \code{@} sign.

\begin{wrapfigure}{l}{0.55\textwidth}
\vspace*{-0.8cm}
\hrule
\vspace*{-0.1cm}
\begin{lstlisting}
void lex()
  _(trace letter+ at letter+ dot letter+ eof)
{
  int state = 0;

  while(0 <= state && state <= 5)
    _(invariant 0 <= state && state <= 6)
    // _(trace ...) see running text
  {
    char c = next();
    check(c);

    if(state == 0) {
      if('a' <= c && c <= 'z') state = 1;
      else                     abort();
    } else if(state == 1) {
      if('a' <= c && c <= 'z') state = 1;
      else if(c == '@')        state = 2;
      else                     abort();
    } // states 2 to 4 omitted
      ...
    } else if(state == 5) {
      if('a' <= c && c <= 'z') state = 5;
      else if(c == -1)         state = 6;
      else                     abort();
    }
  }
}
\end{lstlisting}
\vspace*{-0.4cm}
\caption{E-Mail address matcher in SecC.}
\vspace*{+0.1cm}
\hrule
\vspace*{-0.8cm}
\label{fig:lex}
\end{wrapfigure}
Often, such matchers are implemented as a library in programming languages,
which at runtime translates textual representations of regular expressions into some internal automaton.
For high-performance code, however, the preferred alternative is to compile this automaton upfront into source code, e.g., using the popular tool \texttt{re2c}~\cite{bumbulis1993re2c}.
Possible translation schemes make use of language features like loops/gotos and switch/conditionals, representing the automaton either implicitly via control flow,
or with an outer loop and an explicit state variable,
as shown in the example in \cref{fig:lex}.
We will show how to verify this code with respect to the declarative specification
of the expected matches in terms of the regular expression trace annotation at the top.%
    \footnote{Available in the SecC bitbucket repository at \code{examples/case-studies/matcher.c}}

The input is read by repeated calls to some external function \code{next()}
which returns the next input character.
The switching logic of the state machine is implemented via variable \code{state} which is modified according to transitions encoded in the nested chain of conditionals.
For example, state~0 expects a first letter symbol as the repetition
is non-empty, whereas state~1 may subsequently transition over an occurrence of the symbol~\code{@}.
The loop terminates in state~6 after the input becomes empty, encoded by \code{c == -1}.
To simplify matters further, the code simply calls \code{abort()}
to denote an unexpected input character is encountered,
which exposes nontermination via a postconditon \code{\_(ensures false)},
i.e., a call to this function will vacuously verify that branch of the computation.

The function \code{check()}, shown in \cref{fig:check},
denotes the events that are emitted in relation to the respective input characters,
and also includes an \code{error} event for all other characters.
This additional event does not occur in the top-level specification
of \code{lex()}, such that clearly the only way to satisfy the trace constraints
when such an error is encountered is to reject the match via a call \code{abort()}.
Function \code{check()} is a specification artifact
that separates out the \emph{interpretation} of inputs in terms of the
four abstract events \code{letter}, \code{at}, \code{dot}, and \code{eof}.
While it may appear somewhat cumbersome to draw this connection explicitly,
we emphasize that this just reflects the prototypical nature of the design
that does not ascribe any meaning to events upfront.
Of course one may provide support for certain kinds of events out-of-the-box,
perhaps complemented by character class definitions as abbreviations for finite enumerations.

\begin{wrapfigure}{l}{0.5\textwidth}
\vspace*{-0.8cm}
\hrule
\vspace*{-0.1cm}
\begin{lstlisting}
void check(char c);
  _(trace letter if 'a' <= c && c <= 'z')
  _(trace at     if c == '@')
  _(trace dot    if c == '.')
  _(trace eof    if c == -1)
  _(trace error  otherwise)

\end{lstlisting}
\vspace*{-0.4cm}
\caption{Specifying the correspondence between inputs and abstract events.}
\vspace*{+0.1cm}
\hrule
\vspace*{-0.8cm}
\label{fig:check}
\end{wrapfigure}

The verification consists of two parts:
Annotating the \code{while} loop with conditional trace invariants
that must jointly be preserved over the loop body,
and ensuring that this annotation implies correctness of the outer procedure.
The possible trace prefixes in the respective states are captured as follows.
\begin{lstlisting}[escapechar=!]
_(trace ()                                      if state == 0)
_(trace letter+                                 if state == 1)
_(trace letter+ at                              if state == 2)
  ...
_(trace letter+ at letter+ dot letter+ eof      if state == 6)
\end{lstlisting}
% _(trace letter+ at letter+                      if state == 3)
% _(trace letter+ at letter+ dot                  if state == 4)
% _(trace letter+ at letter+ dot letter+          if state == 5)
They reflect the part of the input that has been matched already,
starting with the empty trace, written as \code{()} in SecC in state~0,
up to the final trace in state~6.
Given the negated loop test at exit together with the invariant,
the last line is the only one with a satisfiable guard,
which is precisely the guarantee we need to satisfy the contract of \code{lex()}.
Preservation over a single iteration over the loop body is briefly discussed with respect to state~1.
Its corresponding trace \code{letter+} is initially established from a single \code{letter} event in the incoming transition from state~0.
Reading the next character in the range \code{[a-z]} in state~1 extends the trace to \code{letter+ letter}
by another \code{letter} event from \code{check()},
which is subsumed by \code{letter+} in state~1 again as required.
On the other hand, reading an \code{@} symbol produces the trace specified
for state~2, which is the one transitioned to at the end of this iteration.
All other transitions work analogously.

\section{Case-Study: VerifyThis Casino Challenge}
\label{sec:casino}

In this section we discuss the Casino case study
that fueled a series of online discussions in the context of the VerifyThis competition.
The case study deals with a Casino that was originally implemented as a smart contract on the Ethereum blockchain,
but which has since been specified, modeled, and verified in a number of different approaches.

The Casino game is offered by an operator to a player who can bet on the
outcome of a coin toss, which is decided upfront but remains hidden in an envelope,
until after the player has placed their bet and the operator decides to resolve.
Cheating by the operator is prevented by encoding the sealed envelope by
publishing the cryptographic hash of a secret number whose last bit determines the winning side of the coin.
If the player wins, they receives double the amount of money that they bet,
otherwise it goes to the operator.
After that, the game can be played again.
For a very nice graphical exposition of the rules by Wolfgang Ahrendt
we refer the reader to \url{https://verifythis.github.io/casino/}.

\begin{figure}[t]
\hrule
\vspace*{0.2cm}
\centering
\includegraphics[width=0.9\textwidth]{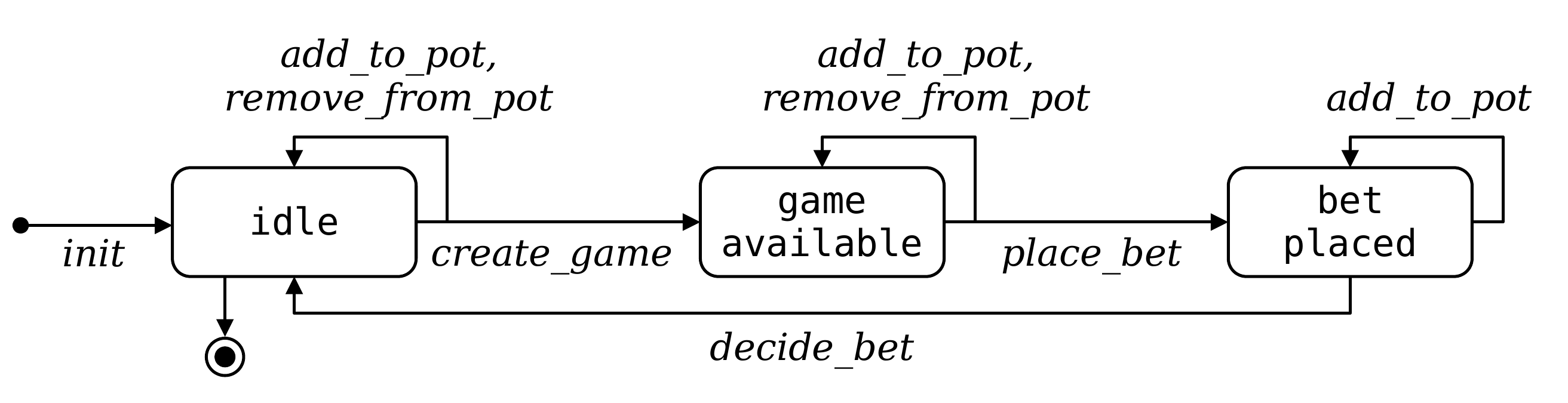}
\caption{Automaton specification (adapted from a model by Matthias Ulbrich, \url{https://verifythis.github.io/casino/spec/} )}
\vspace*{0.2cm}
\hrule
\vspace*{-0.4cm}
\label{fig:casino-statemachine}
\end{figure}
As mentioned in the introduction, a particular goal of the efforts around this challenge was to bridge between different verification approaches,
such as finite state models, contract-based models, and hybrid solutions.
An aspect in focus is the \emph{pot} of money associated with the game,
which is supplied by the operator and which needs to cover the prize of the player,
causing an invariant that encodes this requirement.
\Cref{fig:casino-statemachine} shows a high-level description that avoids
draining the pot by simply restricting the operator
to reduce the amount in the pot as long as an unresolved bet has been placed.

SecC models this game as follows:%
    \footnote{Available in the SecC bitbucket repository at \code{examples/case-studies/casino.c}}
Data on the blockchain is treated as memory that has public visibility,
a feature that is supported by the logic.
All data passed in and out of the operations that encode the moves in the game
are likewise specified to be public.
Ownership, e.g., of a wallet or a payment that has been sent but not received yet,
naturally maps to resources in Separation Logic, encoded into abstract predicates
that cannot be duplicated.

The game itself is specified as a main function with a top-level loop
that nondeterministically chooses
among the next possible moves and then calls the function
that implements the respective transition.
Each of these functions emits one of the events shown in \cref{fig:casino-statemachine}.
The corresponding behavior of the game is captured by the regular expression shown below.
It declaratively specifies all traces through the automaton in \cref{fig:casino-statemachine}.
We point out the inclusion of some trailing changes to the pot.
% which we had omitted in an earlier version, and which was caught
% only after by fixing the subsumption checking algorithm as discussed in \cref{sec:implementation}.
Moreover, it is noteworthy that after a \event{place\_bet} event,
no further \event{remove\_from\_pot} can happen until the game is decided.
\begin{lstlisting}
_(trace init
    ((add_to_pot | remove_from_pot)* create_game
     (add_to_pot | remove_from_pot)* place_bet
      add_to_pot*                    decide_bet)*)
    (add_to_pot | remove_from_pot)* 
\end{lstlisting}
Since the game can be in one of three states (idle, game available, bet placed),
the trace invariant for the loop has three parts, that analogously to those of the matcher in \cref{sec:lex}
reflect the different stages of the game.
Like with the even/odd example, the trace invariant duplicates as a prefix
the entire specification in order to be inductive,
which results in a somewhat large, but structurally straight-forward annotation,
for example that right before deciding the bet we have seen everything \emph{apart} from the event \event{decide\_bet}.
\begin{lstlisting}
_(trace (...)
    ((add_to_pot | remove_from_pot)* create_game
     (add_to_pot | remove_from_pot)* place_bet
      add_to_pot*)  if state == BET_PLACED)
\end{lstlisting}
where \code{(...)} here omits the entire game specification as shown above.

We have experimented further with simplified variants that just include
the functionality directly related to the trace behaviors.%
    \footnote{Available in the repository at \code{examples/case-studies/casino-statemachine.c}}
The approach taken for the full case study is complemented by the naive approach
with rule \textsc{While*}.
Furthermore, we formulate the game as a set of several mutually tail-recursive procedures,
each corresponding to one state.
This affects the annotations in two ways:
First, each function on its own can specify its contribution to the observable behaviors, and second, we are effectively encoding the summary-based approach from \cref{sec:motivation},
which in comparison to an invariant reasons about the trace suffixes that are
yet to be fulfilled, which in comparison to the corresponding part of the loop annotation,
specifies possible additions to the pot and includes a final \code{decide\_bet} before the game repeats recursively (where again \code{(...)} omits the game but without \event{init}).
\begin{lstlisting}
void game_placed()
_(trace (add_to_pot* decide_bet) // <-- behavior of just game_placed
        (...))
\end{lstlisting}

\subsection{Discussion}
\label{sec:discussion}

The two case studies demonstrate how behavioral specifications
in terms of regular expressions can capture declaratively properties of C functions
as part of their top-level modular contracts.
However, loop annotations become quite involved,
and in the approach taken
% At least in the particular approach we have chosen,
the user has to come up with the equivalence classes of regular
expressions corresponding to the control states of the loop.

Moreover, these expressions are typically larger than the original annotation,
and are repeated in significant parts for the many different cases.
We think that this can be addressed for example by allowing
the user to abbreviate expressions as it is possible for example
in typical scanner generators like \code{flex}.
Another idea is to infer such specifications automatically,
e.g., by having the user annotate the \emph{conditions} that partition the state space,
but not the regular expressions themselves.
We think that this is feasible but we leave this idea for future work.

SecC verifies both case studies in less than~5s on a Thinkpad T470p.
We do not think that these numbers are particularly meaningful,
as they were done with a cold-cache JVM and they include many calls to an external SMT solver.
The overhead introduced by the regular expression inclusion check contributes noticeably to the time to verify functions which have many paths,
like the main loops of the case studies.
Nevertheless, this was never a limiting factor here, on the contrary,
thanks to the decidability of the inclusion check once the conditionals are settled,
there is no additional manual effort involved to help out the verifier with additional proof hints,
as it is often typical with expressive functional contracts.
Perhaps, when scaling up to regular expressions as they may occur in practice,
an efficient automata-based implementation may be preferable.

Finally, as hinted at already in \cref{sec:lex},
built-in support for certain kinds of domain-specific events
can help streamline the verification process.
For example, related work has considered opening and closing of file handles~\cite{das2002esp},
as well as function calls and returns~\cite{disney2011temporal}.
These applications, however, would strongly benefit from a context-free specification language, instead of a regular one, to pair each open with a close for instance.
Similarly, attaching data drawn from finite sets or even symbolic data like file handles to events is useful.
However, such extensions reflect different trade-offs wrt.\ automation as the inclusion check may become undecidable.
This opens up a research space in between the work presented here
and the highly expressive approaches, in which one can experiment
with practical heuristics supported by domain-specific proof hints when needed.

% \section{Further Case Studies}
% 
% \begin{itemize}
% \item CDDC input handling for hotkey to switch something
% \item TCP handshake? other protocols
% \item Security access 
% \item verify a parser for a particular regular language (e.g. E-Mail addresses),
%       maybe generated from re2c
% \end{itemize}

%%% Local Variables:
%%% mode: LaTeX
%%% mode: TeX-PDF
%%% TeX-source-correlate-mode: t
%%% TeX-master: "main.tex"
%%% ispell-local-dictionary: "en_GB"
%%% End:

\section{Related Work}
\label{sec:related}

Deductive verification tools like SecC~\cite{ernst:cav2019}
come with strong support for logical specifications,
which would in principle enable to encode regular traces into lists explicitly.
Nevertheless, there are several recent approaches that aim for a more first-class support of specifying behaviors,
which enables one to provide certain guarantees (soundness of compositions, automation):

A common strategy is to encode behaviors into abstract permission-based predicates,
which integrates event histories with expressive logical data types.
For example, VeriFast supports such I/O-specifications and has support
not just for safety properties but also for liveness~\cite{penninckx2019specifying,jacobs2020modular}.
VerCors has a highly elaborate mechanism to capture behaviors as process models~\cite{oortwijn2020abstraction},
which can then be given to a model checker to verify system-wide properties,
based on the earlier work~\cite{blom2015history} that tracks histories of events similarly.
% A noteworthy aspect of this work is that it comes with a fully mechanized soundness proof.
Another work that cuts into this direction is Igloo~\cite{sprenger2020igloo},
in which modularity and composition is a key concern.

Interestingly, all of these approaches reason in the \emph{opposite} direction as we do:
The precondition specifies, which traces are still allowed to occur, in other words,
a forward simulation between the code and the process or recursively defined model of the externally
visible behavior is maintained as part of the auxiliary state.
Our initial design of the logic was in fact based on this idea,
but we did not (yet) succeed to prove end-to-end soundness wrt.~\cref{def:hoare}.
Likely this can be remedied with an additional intermediate backward simulation,
as hinted at in \cite{oortwijn2020abstraction},
but for this work we settled for the design as presented which we think is very clean.
Nevertheless, looking further into this issue is certainly of interest,
in parts because~\cite{sprenger2020igloo} relies on a proof system with precisely this kind of guarantee.

Session types~\cite{huttel2016foundations} are another very active reasearch area:
Here, the interactions between software components are captured as part of the type system.
Similarly, interface automata~\cite{de2001interface} capture possible interactions at the system level.
Temporal contracts have been proposed for a functional language in~\cite{disney2011temporal},
but considering runtime checking only, and giving completeness to ease the specification burden.
Work that looks into temporal properties of smart contracts is~\cite{permenev2020verx}.

%%% Local Variables:
%%% mode: LaTeX
%%% mode: TeX-PDF
%%% TeX-source-correlate-mode: t
%%% TeX-master: "main.tex"
%%% ispell-local-dictionary: "en_GB"
%%% End:

\section{Conclusion}

We have presented a Hoare logic that integrates external behaviors as regular expressions
into deductive verification systems.
This can be used to declaratively specify sequences of events occurring during the execution.
The approach supports loops with control states via conditions
in the specification, such that interesting sequential behavior across multiple iterations can be captured.
There is much room for experimentation
with different trade-offs between automation, ease of specification, and expressiveness,
to be explored in the future.

\bibliographystyle{splncsnat}
\begingroup
  %\microtypecontext{expansion=sloppy}
  \small % ensure correct font size for the bibliography
  \bibliography{references}
\endgroup

\end{document}